\documentclass[a4paper,12pt]{article}
\usepackage{amsmath,amssymb,amsfonts,amsthm}
\newtheorem{theorem}{Theorem}[section]
\newtheorem{lemma}[theorem]{Lemma}

\usepackage{graphicx}
\usepackage{enumerate}

\newcommand{\R}{\mathbb{R}}
\newcommand{\Rt}{\mathbb{R}^3}

\newcommand{\Rn}{\mathbb{R}^n}

\newcommand{\cs}[1]{C^{\infty}_0(#1)}

\newcommand{\Lf}{L'^2_{-5/2}}
\newcommand{\Lo}{L'^2_{-1/2}  }
\newcommand{\Ho}{H'^2_{-1/2}}

\newcommand{\Lfn}[1]{\left\lVert#1\right\rVert_{\Lf}}
\newcommand{\Lon}[1]{\left\lVert#1\right\rVert_{\Lo}}
\newcommand{\Hon}[1]{\left\lVert#1\right\rVert_{\Ho}}
\newcommand{\opx}[1]{\left\lVert#1\right\rVert_{\mathcal{L}(X,Z)}}
\newcommand{\opy}[1]{\left\lVert#1\right\rVert_{\mathcal{L}(Y,Z)}}

\title{Small deformations of extreme Kerr black hole initial data}

\author{Sergio Dain$^{1,2}$ and Mar\'\i a E. Gabach Cl\'ement$^{1}$\\
\\
$^1$Facultad de Matem\'atica, Astronom\'{i}a y F\'{i}sica, FaMAF,\\
Universidad Nacional de C\'ordoba,\\
Instituto de F\'{\i}sica Enrique Gaviola, IFEG, CONICET,\\
 Ciudad Universitaria, (5000) C\'ordoba, Argentina.  \\
$^{2}$Max Planck Institute for Gravitational Physics,\\
  (Albert Einstein Institute), Am M\"uhlenberg 1,\\
  D-14476 Potsdam Germany. }

\begin{document}

\maketitle

\abstract{We prove the existence of a family of initial data for Einstein
  equations which represent small deformations of the extreme Kerr black hole
  initial data.  The data in this family have the same asymptotic geometry as
  extreme Kerr. In particular, the deformations preserve the angular momentum
  and the area of the cylindrical end.}

\section{Introduction}


Black holes are one of the most spectacular predictions of General Relativity.
There is growing experimental evidences that indicate that black holes do
indeed exists in nature. Among the most impressive ones are the evidences for
the existence of a supermassive black hole in the center of our galaxy (see the
review article \cite{Reid:2008fp}).

In vacuum, the only stationary black hole is expected to be the Kerr black
hole, characterized by the mass $m$ and the angular momentum $J$ (see
\cite{Chrusciel:2008js} and reference therein for updated results on this
problem). The Kerr black hole satisfies the inequality $m\geq \sqrt{|J|}$. The
limit case $m= \sqrt{|J|}$ is called the extreme black hole. It represents the
stationary black hole with maximum amount of angular momentum per mass
unit. The extreme limit $\sqrt{|J|} \to m$ is singular because the geometry of
the spacetime changes at the limit. This is somehow to be expected since the
extreme case is the borderline between a black hole and a spacetime with a
naked singularity (i.e. the Kerr solution with $0<m < \sqrt{|J|}$.)

There exists relevant reasons to study extreme black
holes. The first one is that there is good experimental evidences for the
existence of nearly extreme black holes in the universe (see
\cite{mcclintock06} for experimental evidences of a black hole with
$J/m^2>0.98$). Then, it is important to understand the dynamics of black holes
near the extreme limit. The second reason is less clear but, we believe, equally
important. As often happens in physical theories, solutions that arise as
asymptotic limits are simpler than other solutions and they provide useful
insight into the theory. In the set of solutions of Einstein equation, extreme
black holes represent a kind of barrier that divide black holes and naked
singularities. From the pure classical point of view, there are evidences that
extreme black holes have some special properties that make them simpler than
non-extreme ones (see the discussion in \cite{Dain:2007pk}). Also, from a
completely different perspective, namely holographic dualities, particular
features of extreme black holes play an important role (see
\cite{Bardeen:1999px} \cite{Guica:2008mu}, see also the review article
\cite{Physics.2.102}).  It appears that extreme black holes have a deep
mathematical structure that it is still uncover.

Finally, there is a third reason to study extreme black holes. In the study of
extreme black hole initial conditions (which is the subject of this article), a
particular kind of geometry appears:  geometries with cylindrical
ends. This geometries have provided to be very useful in the numerical
computations of black holes collisions (they are called 'trumpet' initial
conditions in this context, see \cite{Hannam:2009ib} \cite{Hannam:2008sg}
\cite{Immerman:2009ns}).

As a first step to understand the dynamics near an extreme Kerr black hole, in
this article we study small deformation of the extreme Kerr black hole initial
conditions. We prove the existence of a family of initial data that are close
to extreme Kerr initial data. In particular, the asymptotic geometry of these
initial data is the same as the extreme Kerr geometry.  These data are,
generically, non-stationary. It is important to emphasize that the existence of
these initial conditions it is a priori by no means obvious due to the
character of the extreme Kerr geometry.

The paper is organized as follows. We begin in section \ref{secpert} with a
review of some of the main properties of the extreme Kerr black hole. Then we
state our main result avoiding technical details. We also discuss how the
cylindrical geometry is preserved along the evolution. In section \ref{prueba}
we state our main theorem in a precise form and prove it.  We conclude the
article with a discussion of some relevant open problems in section
\ref{sec:final-comments}. Finally, we have included three appendices. In appendix
\ref{sec:weight-sobol-spac} we prove a decay property of the Sobolev spaces
used in our proof. In appendix \ref{sec:appendix} we prove a property of the
extreme Kerr initial data that plays a central role in the proof. Appendix
\ref{sec:impl-funct-theor} is brief summary of the implicit function theorem,
which is the central analytical tool used in the proof.

\section{Main Result}\label{secpert}

Consider the Kerr black hole with mass $m$ and angular momentum $J$. In the
non-extreme case (i.e. $m>\sqrt{|J|}$) the maximal analytical extension of the
metric has the well known global structure shown in figure \ref{fig:1} (see
\cite{Boyer67}\cite{Carter:1968rr} and also \cite{Carter73}).  Take the
spacelike surface $S$ drawn in this figure.  This surface runs from one
spacelike infinity (denoted by $i_0$) to the other. The topology of this
surface is $S=\mathbb{S}^2\times \mathbb{R}$.  
\begin{figure}
\begin{center}
\includegraphics[width=9cm]{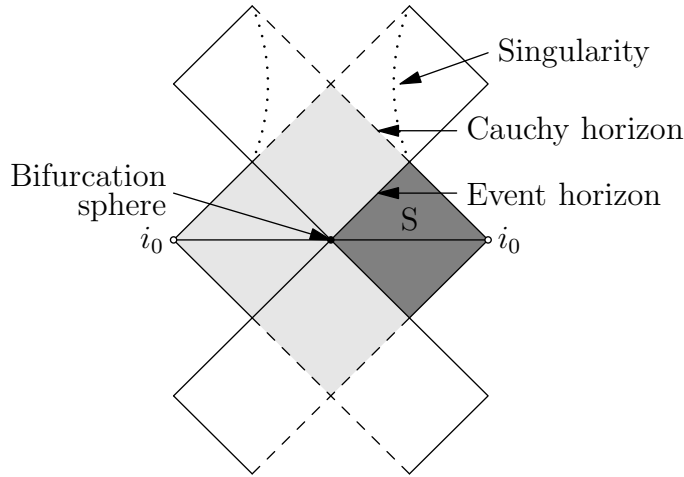}
\caption{Conformal diagram of the Kerr black hole in the non-extreme case. \label{fig:1}}
\end{center}
\end{figure}
The triple $(S,h_{ij}, K_{ij})$, where $ h_{ij}$ is the
induced intrinsic metric on $S$ and $ K_{ij}$ is the second fundamental
form of $S$, constitute an initial data set for Einstein equations. 
That is, they are solutions of the constraint equations
\begin{align}
 \label{const1}
   D_j   K^{ij} -  D^i   K= 0,\\
 \label{const2}
   R -  K_{ij}   K^{ij}+  K^2=0,
\end{align}
where $ {D}$ and $ R$ are the Levi-Civita connection and the Ricci scalar
associated with ${h}_{ij}$, and $ K = K_{ij} h^{ij}$. In these equations the
indices are moved with the metric $ h_{ij}$ and its inverse $ h^{ij}$.

The Riemannian manifold $(S, h_{ij})$ has two asymptotically flat ends (see
figure \ref{fig:2}). This asymptotic geometry is identical to the analogous slice of 
Kruskal extension for Schwarzschild black hole.  The surface $S$ in figure
\ref{fig:1} corresponds to a slice $t=0$ of the Boyer-Lindquist coordinates
$(t,\tilde r,\theta,\phi)$ in Kerr metric (see the appendix
\ref{sec:appendix}). It intersects the bifurcation sphere (denoted by a dark
dot in figure \ref{fig:1} and by a dark circle in figure \ref{fig:2}). The
slice is isometric across this sphere.  The bifurcation sphere on the slice is
both a minimal surface and an apparent horizon.  In these coordinates,
spacelike infinity $i_0$ is represented by the limit $\tilde r\to \infty$. The
intrinsic metric and the second fundamental form satisfy the standard
asymptotically flat fall-off conditions 
\begin{equation}
  \label{eq:7}
  h_{ij}=\delta_{ij}+O(\tilde  r^{-1}), \quad  K_{ij}=O(\tilde r^{-3}),
  \text{ as } \tilde r\to \infty, 
\end{equation}
where $\delta_{ij}$ is the flat metric. 
The strong fall-off behavior of the second fundamental form implies that the
linear momentum of the initial data vanishes. The angular momentum is contained
in the term $O( r^{-3})$ of $ K_{ij}$.  

The maximal development of the initial data set $(S, h_{ij}, K_{ij})$ is shown
in light gray in figure \ref{fig:1}. This region does not cover the whole
analytical extension (as in the case of Schwarzschild's), it has a
smooth boundary in the spacetime. This boundary is known as  Cauchy
horizon. In dark gray the domain of outer communications is shown, which is
bounded by the black hole event horizon.

\begin{figure}
\begin{center}
\includegraphics[width=4cm]{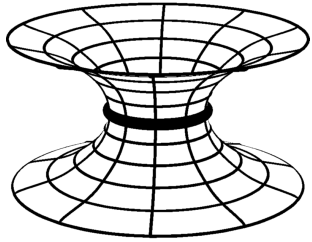}
\caption{The initial data for the non-extreme case. The dark circle in the
  middle represents the minimal surface. \label{fig:2}}
\end{center}
\end{figure}

In the extreme case $m=\sqrt{|J|}$ the global structure of the spacetime
changes. The maximal analytical extension is shown in figure \ref{fig:3}. 
\begin{figure}
\begin{center}
\includegraphics[width=6cm]{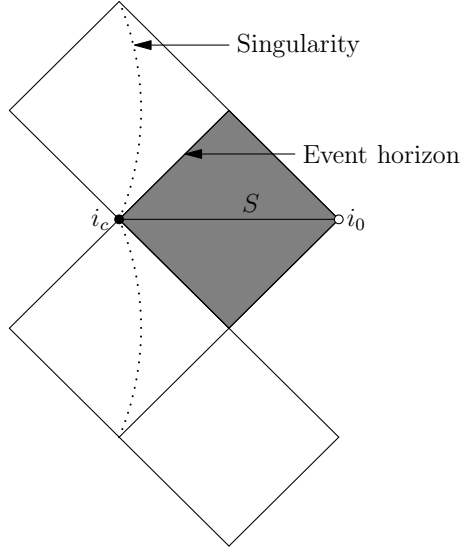}
\caption{Conformal diagram for the extreme Kerr black hole. The cylindrical end
  is denoted by $i_c$. \label{fig:3}}
\end{center}
\end{figure}
The spacelike surface $S$ has the same topology $\mathbb{S}^2\times \mathbb{R}$
as in the non-extreme case, however, the asymptotic geometry of the Riemannian manifold
$(S, h_{ij})$ is different. It has one asymptotically flat end and one
cylindrical end, see figure \ref{fig:4}. The cylindrical end asymptotically approaches
 the event horizon. Contrary to the asymptotically flat case,
this end is in the strong field region of the spacetime.  Note that
$(S,h_{ij})$ is a complete Riemannian manifold without boundary which lies
completely in the black hole exterior region.
\begin{figure}
\begin{center}
\includegraphics[width=4cm]{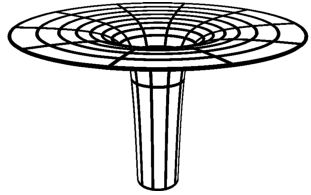}
\caption{The initial data for the extreme Kerr black hole. \label{fig:4}}
\end{center}
\end{figure}
Let us take a closer look at the structure of the cylindrical end. In
isotropic coordinates ($r,\theta,\phi$), with $r:=\tilde r-m$ (see appendix
\ref{sec:appendix}), the induced metric on $S$ has the form
\begin{equation}\label{thcero}
 h^0_{ij}=\Phi_0^4\tilde h^0_{ij},\qquad  \tilde
 h^0=e^{2q_0}(dr^2+r^2d\theta^2)+r^2\sin^2\theta d\phi^2, 
\end{equation}
where $\Phi_0$ and $q_0$ are given by equation \eqref{ficero} in  appendix
\ref{sec:appendix}. The extrinsic curvature is given by 
\begin{equation}
\label{tkcero}
 K^0_{ij}=\frac{2}{\eta} S_{(i} \eta_{j)}, \quad
 S_i=\frac{1}{\eta}\epsilon_{ijk}\eta^j\partial^k\omega_0,
\end{equation}  
where $\eta^i$ is the axial Killing vector, $\eta$ the square of its norm (see
equation \eqref{eq:59}), $\epsilon_{ijk}$ denotes the volume element with
respect to the metric $h_{ij}$ and $\omega_0$ is given by \eqref{wcero}. The
advantage of this particular form of writing $K^0_{ij}$ is that it is easy to
check from \eqref{tkcero} that $K^0_{ij}$ satisfies the momentum constraint
\eqref{const1}.  We will discuss and use this fact in section \ref{prueba}. In
particular, we have that $K^0_{ij}$ is trace-free
\begin{equation}
  \label{eq:62}
   K^0=0.
\end{equation}
That is, these initial data are maximal surfaces. 

In isotropic coordinates, the asymptotically flat end is given by the limit $r\to
\infty$ and the cylindrical end by the limit $r\to 0$.  The radial coordinate
$r$ is a good coordinate in the asymptotically flat end since the metric and
the extrinsic curvature take the asymptotic form \eqref{eq:7}.

On the other hand, in the limit $r \to 0$ the conformal factor $\Phi_0$ blows
up. This is, however, just a coordinate problem. To see this, let $s=-\ln r$, then the
cylindrical end corresponds to $s\to \infty$, and the metric has the
form
\begin{equation}
  \label{eq:2}
    h^0=(\sqrt{r}\Phi_0)^4\left(e^{2q_0}(ds^2+d\theta^2)+\sin^2\theta
     d\phi^2 \right).
\end{equation}
The functions $\sqrt{r}\Phi_0$ and $q_0$ are smooth and uniformly bounded in
the whole range $-\infty < s< \infty$ (see lemma \ref{desi}). In particular, the
Riemannian manifold $(S, h^0_{ij})$ has bounded curvature.

It is interesting to note (although we will not make use of it) that the metric
\eqref{eq:2} and the second fundamental form \eqref{tkcero} have a well defined
limit $s\to \infty$ as initial data. Namely
\begin{equation}
  \label{eq:55}
 h^0=m^2(1+\cos^2\theta) \left(ds^2+d\theta^2\right)
 +\frac{4m^2\sin^2\theta}{(1+\cos^2\theta)}  d\phi^2, \text{ as } s\to \infty, 
\end{equation}
where we have used the limits \eqref{eq:56}--\eqref{eq:57}. The extrinsic
curvature $K^{ij}_0$ has the form \eqref{tkcero} where $\omega_0$ is replaced
by its limiting value \eqref{eq:58} and all the other quantities are computed
with respect to the metric \eqref{eq:55}.  These are in fact solutions of the
constraint equations \eqref{const1}--\eqref{const2}. They isolate the
cylindrical geometry cutting off the asymptotically flat end. In particular,
the metric \eqref{eq:55} has non-negative Ricci scalar, given by the limit
\eqref{eq:58b} and it has another symmetry, namely translations in $s$.  These
limit initial data are slices $t=constant$ of the four dimensional vacuum
geometry described in \cite{Bardeen:1999px}, known as the the near-horizon
extreme Kerr. This geometry has also been studied in \cite{Carter73} (see
eq. (5.63) in that reference).

A relevant parameter for extreme black hole data is the area of the cylindrical
end. Consider the area $A(r)$ of the surfaces $r=constant$ of the metric
\eqref{thcero}. In the limit $r\to 0$ we have
\begin{equation}
  \label{eq:3}
  A_0= \lim_{r\to0}A(r)=8\pi m^2.
\end{equation}
For extreme Kerr, this corresponds to the area of the black hole event horizon.
Finally, for completeness, let us mention that the ergoregion on $S$ is given
in these coordinates by
\begin{equation}
  \label{eq:6}
  0<r < m\sin\theta.
\end{equation}

We have described a particular class of initial data sets for the extreme Kerr
black hole which run from $i_c$ to $i_0$. There exist similar initial data sets
in Reissner-Nordstr\"om and Kerr-Newman black hole.  Remarkably enough, for a
Schwarzschild black hole there also exist initial data that having the same
asymptotic geometry (see \cite{Hannam:2008sg} and references therein).  All
these examples are stationary.  Moreover, all these data arise as a singular
limit in which the geometry changes.  The first numerical evidence for the
existence of non-stationary cylindrical data with a similar structure as the
one described above was given in \cite{Dain:2008ck} and the first analytical
proof was provided in \cite{Dain:2008yu}, \cite{gabach09}. These data are also
obtained as a singular limit from non-extreme data. The point we want to
address to in this article is the following: given extreme Kerr initial data,
does there exist a neighborhood of similar data? The following theorem, which
constitutes the main result of this article, gives an affirmative answer to
this question.

\begin{theorem}
\label{t:family}
  Let $(S, h^0_{ij}, K^0_{ij})$ be the extreme Kerr data set described above
  with angular momentum $J$ and mass $m=\sqrt{|J|}$. Then there is a small
  $\lambda_0>0$ such that for $-\lambda_0<\lambda <\lambda_0$ there exists a family
  of initial data sets $(S, h_{ij}(\lambda), K_{ij}(\lambda))$ (i.e.,
  solutions of the constraints on $S$)
  with the following properties:
\begin{enumerate}[(i)]
\item We have $h_{ij}(0)= h^0_{ij}$ and $K_{ij}(0)=K^0_{ij}$. The family is
  differentiable in $\lambda$ and it is close to extreme Kerr with respect to
  an appropriate norm which involves two derivatives of the metric.

\item The data has the same asymptotic geometry as the extreme Kerr initial
  data set. The angular momentum and the area of the cylindrical end in the
  family do not depend on $\lambda$, they have the same value as in $(S,
  h^0_{ij}, K^0_{ij})$, namely $J$ and $8\pi |J|$ respectively.

\item The  data are axially symmetric and maximal (i.e. $K(\lambda)=0$). 

\end{enumerate}

\end{theorem}
In section \ref{prueba} we provide a more precise version of this theorem
(theorem \ref{teo}). Let us discuss here other relevant properties of the
initial data family $(S, h_{ij}(\lambda), K_{ij}(\lambda))$.

We mention that the angular momentum of the family remains constant, the total
mass however is not. As a consequence of the general theorems \cite{Dain06c}
\cite{Chrusciel:2007ak} we have the following inequality for all $\lambda$
\begin{equation}
  \label{eq:1}
  m(\lambda)\geq \sqrt{|J|},
\end{equation}
with equality only for $\lambda=0$ (i.e. for extreme Kerr). 
This family realizes the local minimum behavior of extreme Kerr studied in
\cite{Dain05d}. 

Inequality \eqref{eq:1} allows us to define the following positive quantity
\begin{equation}
  \label{eq:54}
  E(\lambda)= m(\lambda)- \sqrt{|J|}.
\end{equation}
The energy $E$ provides (if we assume cosmic censorship) an upper bound for the
total amount of radiation emitted by the system at null infinity for these 
initial data (see the discussion in \cite{Dain:2007pk}).

Let us consider now some aspects of the evolution of these data.  In the
asymptotically flat case, it is well know that the asymptotic behavior
\eqref{eq:7} is preserved by  evolution if we impose appropriate fall off
conditions for the lapse and shift. This is of course important, since it is
related to  conservation of  total mass in the spacetime. The natural
question is whether this kind of persistence under evolution also holds for the
cylindrical asymptote. To study this question we need  non-stationary data
as the ones constructed here. 

Let us consider a member of the family for some $\lambda \neq 0$ (we will
suppress the $\lambda$ in the notation in the following). Take a  
 short period of time $t$, then we have
\begin{align}
 h_{ij}(t) & \approx h_{ij}(0)+\dot h_{ij}(0)t.\\
 K_{ij}(t) & \approx K_{ij}(0)  +\dot K_{ij}  t,
\end{align}
where dot denotes time derivative. The time derivatives $\dot h_{ij},\dot
K_{ij}$ can be computed using the evolution equations
\begin{align}
\label{eq:8}
 \dot  h_{ij} &=2\alpha  K_{ij}+\mathcal L_{\beta}  h_{ij},\\
\label{eq:9} 
\dot  K_{ij}  &=\nabla_i\nabla_j\alpha+\mathcal L_{\beta}
 K_{ij}+\alpha(2 K^{k}_i K_{jk}- K  K_{ij}- R_{ij}), 
\end{align}
where $\alpha$ and $\beta^i$ are the lapse and shift of the foliation,
$\mathcal L$ denotes the Lie derivative and $ R_{ij}$ is the Ricci tensor of
$h_{ij}$.  If we want to preserve the cylindrical geometry under the evolution,
we must have
\begin{equation}
  \label{eq:4}
  \lim_{s\to \infty} \dot h_{ij}=0, \quad  \lim_{s\to
    \infty} \dot K_{ij}=0. 
\end{equation}
From equations \eqref{eq:8}--\eqref{eq:9} we deduce the following conditions
for  the lapse
\begin{equation}
  \label{eq:5}
\lim_{s\to \infty}  \alpha= \lim_{s\to \infty}  \partial \alpha=\lim_{s\to
  \infty} \partial^2 \alpha =0.  
\end{equation} 
and the shift 
\begin{equation}
  \label{eq:10}
 \lim_{s\to \infty} \beta^i = \lim_{s\to \infty}  \partial \beta^i  =0,
\end{equation}
where $\partial$ denotes partial derivatives with respect to the space
coordinates. 
Note that for the particular Boyer-Lindquist foliation in extreme Kerr these
requirements are satisfied (see equations \eqref{eq:65}--\eqref{eq:66} in
appendix \ref{sec:appendix}).  Conditions \eqref{eq:5} and \eqref{eq:10} are
analogous to the asymptotically flat conditions for lapse and
shift.

In this article, we have assumed vacuum for simplicity. We expect that an analogous
result as theorem \eqref{thcero} holds for the Kerr-Newman extreme black
hole. In that case, inequality \eqref{eq:1} should be replaced by it
generalized charged version recently proved in \cite{Chrusciel:2009ki}, 
\cite{Costa:2009hn}.   

\section{Proof of Main Result}
\label{prueba}
A particular feature of axial symmetry is that it allows one to  reduce the constraint
equations \eqref{const1}--\eqref{const2} to just one scalar equation for a conformal
factor (the so called Lichnerowicz equation). This procedure is well known
(see, for example, \cite{Dain06c} and reference therein).  Let us briefly
review it. Consider the metric
\begin{equation}
  \label{eq:105}
\tilde  h_{ij} = e^{-2q}(dr^2+r^2d\theta^2)+ r^2\sin^2\theta \,  d\varphi^2,
\end{equation}  
where $q=q(r,\theta)$ is an arbitrary function. 
This metric will be used as a conformal background for the physical metric $
h_{ij}$.  We first discuss how to construct solutions of the momentum
constraint \eqref{const1} from an arbitrary axially symmetric potential
$\omega(r,\theta)$. Consider
the following tensor
\begin{equation}
  \label{eq:42b}
  \tilde K^{ij}= \frac{2}{\rho^2} \tilde S^{(i} \eta^{j)}, 
\end{equation}
where 
\begin{equation}
  \label{eq:s}
 \tilde S^i=\frac{1}{2\rho^2}\tilde \epsilon^{ijk}\eta_j \partial_k \omega,
\end{equation}
and $\tilde \epsilon_{ijk}$ denotes the volume element with respect to $\tilde
h_{ij}$, $\tilde D$ is the connexion with respect to $\tilde h_{ij}$ and
$\rho=r\sin\theta$ is the cylindrical radius. The indices on tilde quantities
are moved with $\tilde h_{ij}$ and its inverse $\tilde h^{ij}$. The tensor
$\tilde K^{ij}$ is symmetric, trace free, and satisfies the following equation
(see, for example, the appendix in \cite{Dain99b})
\begin{equation}
  \label{eq:dk}
\tilde   D_i\tilde K^{ij}=0, 
\end{equation} 
for arbitrary $q$ and $\omega$. Equation \eqref{eq:dk} essentially solves (up
to a conformal factor) the momentum constraint \eqref{const1}. 
Assume that we have a solution $\Phi$ of Lichnerowicz equation
\begin{equation}
  \label{eq:Lich}
  \Delta_{\tilde h} \Phi -\frac{\tilde R}{8} \Phi=- \frac{\tilde K_{ij}\tilde
    K^{ij}}{8\Phi^7},  
\end{equation}
where  $\Delta_{\tilde h}$ is the Laplacian with respect to $\tilde h_{ij}$ and
$\tilde R$ is the Ricci scalar of  $\tilde h_{ij}$. Consider the rescaling
\begin{equation}
  \label{eq:61}
  h_{ij}=\Phi^4\tilde h^{ij},  \quad K_{ij}=\Phi^{-2} \tilde K_{ij}.
\end{equation}
Then, as a consequence of \eqref{eq:dk} the pair $(h_{ij}, K_{ij})$ satisfies
the constraints \eqref{const1}--\eqref{const2}. That is, the problem reduces to
solving equation \eqref{eq:Lich}. This equation can be written in the following
remarkably simple form in axial symmetry
\begin{equation}
  \label{eq:60}
   \Delta \Phi=-\frac{(\partial\omega)^2}{16\rho^4\Phi^7}-\frac{\Delta_2q}{4}\Phi,
\end{equation}
where, $\Delta$ and $\Delta_2$ are flat Laplace operators in three and two
dimensions respectively (see \eqref{eq:67}).  
In particular, extreme Kerr initial data satisfies this equation, namely
\begin{equation}
\label{eckerr}
 \Delta \Phi_0=-\frac{(\partial
   \omega_0)^2}{16\rho^4\Phi_0^7}-\frac{\Delta_2q_0}{4}\Phi_0.
\end{equation}
The idea is to perturb equation \eqref{eq:60} around the extreme Kerr
solution by taking
\begin{equation}
  \label{eq:63}
  q_0+\lambda q, \quad \omega_0 + \lambda \omega,
\end{equation}
for some fixed functions $q$ and $\omega$ and small $\lambda$, and then to
find a solution $u$ defined by
\begin{equation}
  \label{eq:64}
  \Phi =\Phi_0+u.
\end{equation}
Inserting \eqref{eq:63} and \eqref{eq:64} in equation \eqref{eq:60} and using
 \eqref{eckerr} we obtain our final equation
\begin{equation}
  \label{eq:12}
  G(\lambda,u)=0,
\end{equation}
where we have defined
\begin{equation}\label{eq}
G(\lambda,u)=\Delta u+\frac{(\partial w_0+\lambda\partial
  w)^2}{16\rho^4(\Phi_0+u)^7}-\frac{\partial
  w_0^2}{16\rho^4\Phi_0^7}+\lambda\frac{\Delta_2q}{4}(\Phi_0+u)+\frac{\Delta_2q_0}{4}u.
\end{equation}
Then, theorem \ref{t:family} is a direct consequence of the following existence
theorem for equation \eqref{eq:12}. 
\begin{theorem}\label{teo}
Let $w\in\cs{\Rt\setminus\Gamma}$ and $q\in\cs{\Rt\setminus \{0\}}$. Then,
there is $\lambda_0>0$ such that for all $\lambda\in (-\lambda_0, \lambda_0)$  there
exists a solution $u(\lambda)\in H^{'2}_{-1/2}$ of equation
\eqref{eq:12}. The solution $u(\lambda)$ is continuously differentiable in
$\lambda$ and it satisfies $\Phi_0+u(\lambda)>0$. Moreover, for small $\lambda$
and small $u$ (in the norm $ H^{'2}_{-1/2}$) the solution $u(\lambda)$ is the
unique solution of equation \eqref{eq:12}.   
\end{theorem}
We have used the following notation: $\Gamma$ denotes the axis $\rho=0$,
$\cs{\Omega}$ are smooth functions with compact support in $\Omega$ and $
H^{'2}_{-1/2}$ denotes the Sobolev weighted spaces defined in appendix
\ref{sec:weight-sobol-spac}.

The fact  $\omega$ vanishes at the axis implies that the angular momentum
remains fixed for the whole family (see the discussion in \cite{Dain06c}). Also,
using lemma \ref{decaimiento}, from $u\in   H^{'2}_{-1/2}$ it follows that the
perturbation $u$ does not change the area of the cylindrical end at $r=0$. 

\begin{proof} 
  The proof uses the Implicit Function Theorem (see theorem \ref{t:ift} in
  appendix \ref{sec:impl-funct-theor}, in the rest of the proof we will follow
  the notation introduced in that theorem) for the map $G$ defined in equation
  \eqref{eq}. The proof is divided in two steps.

  In the first step, we find the appropriate Banach spaces $X$, $Y$ and $Z$
  required by theorem \ref{t:ift}, together with the neighborhoods $U\subset X$
  and $V\subset Y$, such that $G:V\times U\to Z $ defines a $C^1$ map.  The
  delicate part of this step is to take into account in the definition of the
  Banach spaces the fall off behavior at infinity and the singular behavior at
  the origin of the background functions $\Phi_0$, $q_0$ and $\omega_0$.  In
  particular, it is clear from the equation that we can not expect the solution
  $u$ to be regular at the origin, and hence standard Sobolev spaces are not
  appropriate.  Also, the presence of the singular background functions
  $\Phi_0$, $q_0$ and $\omega_0$ in the map $G$ prevents one from using
  standard theorems (for example the chain rule in Sobolev spaces) to prove
  that $G$ is $C^1$.  We need to explicitly compute the functional partial
  derivatives from their very definition as a limit. This makes this part of
  the proof laborious.  The asymptotic behavior of the background Kerr's
  functions is typical of any data with one asymptotically flat end and one
  cylindrical end and that is the main ingredient needed in this step.
  
  In the second step, we prove that the derivative $D_2G(0,0)$ is an isomorphism
  between $Y$ and $Z$. In this part we use very specific properties of extreme
  Kerr initial data (namely, lemma \ref{lemaalpha}) which are not valid for
  generic cylindrical data. See the comment after the proof of lemma
  \ref{lemaalpha}. This step represents the key part of the proof.

  \textbf{Step 1.} To handle both the fall off behavior at infinity and the
  singular behavior at the origin of the functions  $\Phi_0$, $q_0$ and
  $\omega_0$ we will make use of  weighted Sobolev spaces defined in 
  appendix \ref{sec:weight-sobol-spac}.  We choose $X=\R$, $Y=H'^2_{-1/2}$ and
  $Z= L^{'2}_{-5/2}$.  We also choose $U=\R$. It is clear that the map $G$ is only
  defined when $\Phi_0+u>0$. Hence, we need to find an appropriate
  neighborhood $V$ of $0$ in the Banach space $Y$ such that this condition is
  satisfied. Let us consider $V$ given by the open ball
\begin{equation}
  \label{eq:13}
  ||u||_ { H^{'2}_{-1/2}}<\xi,
\end{equation}
where the  constant $\xi$ is computed as follows. 
From lemma \ref{decaimiento} we have that for $u\in V$ 
\begin{equation}
  \label{eq:14}
  \sqrt{r}|u|\leq C_0\xi,
\end{equation}
where the constant $C_0$ is a Sobolev constant independent of $u$. By lemma 
\ref{desi} we have 
\begin{equation}
  \label{eq:15}
  \sqrt{r}\Phi_0\geq\sqrt{m}. 
\end{equation}
Then, if we choose $\xi$ such that
\begin{equation}
  \label{eq:16}
  \frac{\sqrt{m}}{C_0}>\xi>0,
\end{equation}
we have that for all $u\in V$
\begin{equation}
  \label{eq:17}
  \sqrt{r}(\Phi_0+u)\geq \sqrt{m}-C_0 \xi> 0. 
\end{equation}
The constant $\xi$ will remain fixed for the rest of the proof. 

We first  prove that $G:\mathbb{R}\times V\to  L^{'2}_{-5/2}
$  is well defined as a map. That is, we need to check that for $\lambda \in
\mathbb{R}$ and $u\in V$  we obtain $G(\lambda ,u)\in L^{'2}_{-5/2}$.
Let us compute the norm $L^{'2}_{-5/2}$ of $G(\lambda ,u)$.  Using the
definition \eqref{eq} and the triangle inequality we get
\begin{multline}\label{cotaG}
  \Lfn{G(\lambda ,u)}\leq \Lfn{\Delta u} + \Lfn{\frac{\lambda \partial\omega
      \left(2\partial \omega_0 +\lambda \partial
        \omega\right)}{16\rho^4(\Phi_0+u)^7}}+\\
  +  \frac{\lambda}{4}\Lfn{(\Phi_0+u)\Delta_2q}+\\
  +\Lfn{\frac{(\partial
      \omega_0)^2}{16\rho^4}\left[\frac{1}{(\Phi_0+u)^7}-\frac{1}{\Phi_0^7}\right]}+
  \frac{1}{4}\Lfn{u\Delta_2q_0}.
\end{multline}
From the definition of the $H^{'2}_{-1/2}$-norm it is clear that the first term in the
right-hand side of \eqref{cotaG} is bounded. For the second and third terms we
use the hypothesis that $\omega$ has compact support outside the axis and $q$
compact support outside the origin together with the lower bound \eqref{eq:17}
to conclude that these terms are also bounded.  The delicate terms are the
last two.

For the fourth term we proceed as follows. 
Using the following elementary identity  for real numbers $a$ and $b$
\begin{equation}
\label{eq:ei}
  \frac{1}{a^p}-\frac{1}{b^p}=(b-a)
  \sum_{i=0}^{p-1}a^{i-p}b^{-1-i}.
\end{equation}
we find that 
\begin{equation}
\label{eq:19}
r^{-4}\left(\frac{1}{\Phi_0^7}-\frac{1}{(\Phi_0+u)^7}\right)=uH,
\end{equation}
where $H$ is given by 
\begin{equation}
  \label{eq:18}
  H=\sum_{i=0}^6(\sqrt{r}(\Phi_0+u))^{i-7}(\sqrt{r}\Phi_0)^{-1-i}.
\end{equation}
Using  inequalities \eqref{eq:15} and \eqref{eq:17}  we obtain 
\begin{equation}
  \label{eq:20}
  H\leq C,
\end{equation}
where the constant $C$ depends only on the mass parameter $m$ of the background
extreme Kerr solution. In the following we will generically denote by $C$
constants depending at most on $m$. Then, we have
\begin{align}
\label{eq:21}
 \Lfn{\frac{(\partial
     \omega_0)^2}{\rho^4}\left[\frac{1}{(\Phi_0+u)^7}-\frac{1}{\Phi_0^7}\right]} & \leq
 \Lfn{\frac{C}{r^6}(r^4u H)}\\
&=C\Lon{u}\leq C\Hon{u}.
\end{align}
Where we have used the bound \eqref{eq:cotaw0} in Lemma
\ref{desi} to bound the factor with $\omega_0$ in the first inequality in
\eqref{eq:21}. The last inequality  in \eqref{eq:21} comes from the definition of the
weighted Sobolev space $\Ho$.

For the fifth term, which involves  $q_0$, we use the bound \eqref{cotaqcero}
in lemma \ref{desi}, to find
\begin{equation}
\label{eq:22}
 \Lfn{u \Delta_2q_0}\leq
 C\Lfn{\frac{u}{r^2}}=
 C\Lon{u}\leq  C\Hon{u}. 
\end{equation}
These computations show that all norms involved in
$\Lfn{G(\lambda,u)}$ are finite, hence $G:\R\times V\to\Lf$ is a well defined map.

We will now prove  that $G$ is $C^1$ between the mentioned Sobolev spaces. Let
us denote by $D_1G(\lambda,u)$ the partial Fr\'echet derivative of $G$ with
respect to the first argument evaluated at $(\lambda,u)$ and by
$D_2G(\lambda,u)$ the partial derivative with respect to the second
argument. By definition, the partial derivatives are linear operators between
the following spaces 
\begin{align}
  \label{eq:23}
D_1G(\lambda,u) &:\mathbb{R}\to \Lf, \\ 
D_2G(\lambda,u) & :\Ho\to \Lf.   
\end{align}
We use the notation $D_1G(\lambda,u)[\gamma]$ to denote the operator
$D_1G(\lambda,u)$ acting on $\gamma \in \mathbb{R}$. That is,
$D_1G(\lambda,u)[\gamma]$ defines a function on $\Lf$. In the same way we
denote by $D_2G(\lambda,u)[v]$ the operator acting on a function $v\in\Ho $. 
 
We propose as candidates for these partial derivatives the following linear operators
\begin{align}
D_1 G(\lambda,u)[\gamma] &=\left(\frac{2(\partial w_0+\lambda\partial w)\cdot \partial
  w}{16\rho^4(\Phi_0+u)^7}+\frac{\Delta_2 q}{4}(\Phi_0+u)\right)\gamma,\label{eq:pd1} \\
D_2 G(\lambda,u)[v] & =\Delta v + \left(-\frac{7(\partial w_0+\lambda\partial
  w)^2}{16\rho^4(\Phi_0+u)^8}+\lambda\frac{\Delta_2 q}{4}+\frac{\Delta_2
  q_0}{4}\right)v. \label{eq:pd2}
\end{align}
These operators arise by taking formally the following directional derivatives
to the map $G$
\begin{align}
  \label{eq:24}
  \frac{d}{dt}G(\lambda+t\gamma, u)|_{t=0} &=D_1 G(\lambda,u)[\gamma],\\
 \frac{d}{dt}G(\lambda, u+tv)|_{t=0} &=D_2 G(\lambda,u)[v].
\end{align}

To prove that the map $G:\R\times V \to Z$ is $C^1$ we need to prove the
following items:
 
\begin{enumerate}[(i)]

\item The linear operators \eqref{eq:pd1} and  \eqref{eq:pd2} are
  bounded, namely
  \begin{align}
    \label{eq:25}
     \Lfn{D_1 G(\lambda,u)[\gamma]} &\leq C_1 |\gamma|,\\
 \Lfn{D_2 G(\lambda,u)[v]} &\leq C_2 \Hon{v},
  \end{align}
where the constants $C_1$ and $C_2$ do not depend on $\gamma$ and $v$ respectively. 

\item The operators \eqref{eq:pd1} and \eqref{eq:pd2} are continuous in 
   $(\lambda,u)$ with respect to the operator norms. That is, for every $\delta>0$ there
  exists $\epsilon>0$ such that
  \begin{equation}
    \label{eq:26}
    |\lambda_1-\lambda_2|<\epsilon \Rightarrow  \opx{D_1 G(\lambda_1,u)- D_1
      G(\lambda_2,u)}<\delta,
  \end{equation}
and
\begin{equation}
  \label{eq:27}
  \Hon{u_1-u_2} <\epsilon \Rightarrow  \opy{D_1 G(\lambda,u_1) - D_1
    G(\lambda_2,u_2)}<\delta,  
\end{equation}
where the operator norms used in the right hand side of this inequalities are
defined in appendix \ref{sec:impl-funct-theor}. 

\item  The operators \eqref{eq:pd1} and \eqref{eq:pd2} are the partial Fr\'echet
  derivatives of $G$ (see the definition in appendix \ref{sec:impl-funct-theor}). That is
\begin{equation}
\label{eq:frechet1}
\lim_{\gamma\to0}\frac{\Lfn{G(\lambda+\gamma,u)-G(\lambda,u)-D_1
  G(\lambda,u)[\gamma]}}{|\gamma|}=0, 
\end{equation}
and
\begin{equation}
\label{eq:frechet2}
\lim_{v\to0}\frac{\Lfn{G(\lambda,u+v)-G(\lambda,u)-D_2G(\lambda,u)[v]}}{\Hon{v}}=0. 
\end{equation}
\end{enumerate}

By performing similar computations as above it is straightforward to prove (i) and
also the following estimate
\begin{equation}
  \label{eq:28}
   \Lfn{D_1 G(\lambda_1,u) - D_1 G(\lambda_2,u)}\leq C |\lambda_1-\lambda_2|,
\end{equation}
where $C$ does not depend on $\lambda_1$ and $\lambda_2$. From  inequality
\eqref{eq:28} the continuity with respect to $\lambda$ follows, equation 
\eqref{eq:26} of item (ii). In fact, estimate \eqref{eq:28} is a bit
stronger since it gives uniform continuity.

Continuity in the $u$ direction is more delicate. Using again the
identity \eqref{eq:ei} we have 
\begin{equation}
  \label{eq:29}
r^{-9/2} \left( \frac{1}{(\Phi_0+u_1)}-\frac{1}{(\Phi_0+u_2)}\right)=(u_2-u_1)H,
\end{equation}
where
\begin{equation}
  \label{eq:30}
  H=\sum_{i=0}^7(\sqrt{r}(\Phi_0+u_1))^{i-8}(\sqrt{r}(\Phi_0+u_2))^{-1-i}.
\end{equation}
Using that $u_1,u_2\in V$ and the lower bound \eqref{eq:17} we obtain 
\begin{equation}
  \label{eq:31}
  H\leq C.
\end{equation}
We use the upper bound  \eqref{eq:cotaw0}, together with \eqref{eq:31} to find
\begin{equation}
  \label{eq:32}
 \Lfn{D_1 G(\lambda,u_1) - D_1 G(\lambda_2,u_2)}\leq C \Lfn{\frac{v(u_1-u_2)}{r^{3/2}}}.   
\end{equation}
We bound the right hand side of \eqref{eq:32} as follows
\begin{align}
  \label{eq:33}
  \Lfn{\frac{v(u_1-u_2)}{r^{3/2}}} & = \left(\int_{\Rt} \frac{v^2(u_1-u_2)^2}{r}\,
    dx\right)^{1/2}, \\
&= \left(\int_{\Rt} \frac{(\sqrt{r}v)^2(u_1-u_2)^2}{r^2}\, 
    dx\right)^{1/2} \label{eq:33b}\\
& \leq C \Hon{v}\left(\int_{\Rt} \frac{(u_1-u_2)^2}{r^2}\,
    dx\right)^{1/2}\label{eq:33c}  \\
&  \leq C \Hon{v} \Hon{u_1-u_2}. \label{eq:33d}
\end{align}
Equation \eqref{eq:33} is just  the definition of the $\Lf$-norm and  equation
\eqref{eq:33b} is a trivial rearrangement of factors. The crucial inequality is
\eqref{eq:33c} where we have used lemma \ref{decaimiento}. Finally, line
\eqref{eq:33d} trivially follows from the definition of $\Ho$-norms. Hence, we obtain
our final inequality
\begin{equation}
  \label{eq:34}
 \Lfn{D_1 G(\lambda,u_1) - D_1 G(\lambda_2,u_2)}\leq C    \Hon{v} \Hon{u_1-u_2}.
\end{equation}
From this inequality, the continuity \eqref{eq:27} follows.  

We now prove (iii). The first limit \eqref{eq:frechet1} is straightforward. The
delicate part is the second limit \eqref{eq:frechet2}. We will follow a similar
argument as in the previous calculation. We first compute
\begin{multline}
  \label{eq:36}
G(\lambda,u+v)-G(\lambda,u)-D_2G(\lambda,u)[v]=\\
\frac{(\partial \omega_0+\lambda\partial \omega)^2}{16\rho^4}
\left(\frac{1}{(\Phi_0+u+v)^7}-\frac{1}{(\Phi_0+u)^7}+\frac{7v}{(\Phi_0+u)^8}\right) 
\end{multline}

We have
\begin{equation}
  \label{eq:37}
  r^{-9/2}\left(\frac{1}{(\Phi_0+u+v)^7}-\frac{1}{(\Phi_0+u)^7}+
\frac{7v}{(\Phi_0+u)^8}\right)=v^2H,      
\end{equation}
with
\begin{equation}
  \label{eq:39}
  H=\frac{1}{(\sqrt{r}(\Phi_0+u+v))^7(\sqrt{r}(\Phi_0+u))^8}
  \sum_{\substack{i+j+k=6\\ i,j,k\geq 0}} C_{ijk}(\sqrt{r}\Phi_0)^i
  (\sqrt{r}u)^j (\sqrt{r}v)^k, 
\end{equation}
where $C_{ijk}$ are numerical constants. To bound $H$ we use 
 the upper and lower bounds for $\Phi_0$ given by \eqref{cotasficero} and
the fact that $u,v\in V$ (and hence they satisfy the  bound \eqref{eq:14}).  We
obtain
\begin{equation}
  \label{eq:40}
  |H|\leq C \frac{(r+m)^{6/2}}{(\sqrt{(r+m)}-C_0\xi  )^{15}}\leq C.
\end{equation}
Then, we have
\begin{align}
  \label{eq:36b}
\Lfn{G(\lambda,u+v)-G(\lambda,u)-D_2G(\lambda,u)[v]} &\leq
C\Lfn{\frac{r^{9/2}v^2H}{r^6}},\\
&=\Lfn{\frac{v^2}{r^{3/2}}}. \label{eq:36c}
\end{align}
Using the same argument as we used in equation \eqref{eq:33}--\eqref{eq:33d} we
finally get the desired estimate
\begin{equation}
  \label{eq:42}
\Lfn{G(\lambda,u+v)-G(\lambda,u)-D_2G(\lambda,u)[v]} \leq C \left(\Hon{v}\right)^2.  
\end{equation}
From \eqref{eq:42} it follows \eqref{eq:frechet2}.

\emph{Step 2.} We will prove that $D_2G(0,0):H'^{2}_{-1/2}\to L'^{2}_{-5/2}$ is
an isomorphism.  We write this linear operator in the following form
\begin{equation}
\label{eq:43}
D_2G(0,0)[v]=\Delta v -\alpha v,
\end{equation}
where
\begin{equation}\label{alpha}
\alpha=7\frac{(\partial\omega_0)^2}{16\rho^4\Phi_0^8}-\frac{\Delta_2q_0}{4}.
\end{equation}
By lemma \ref{lemaalpha} we have that $\alpha=h r^{-2}$ where $h$ is a positive
and bounded function in $\Rt$. In \cite{gabach09} it has been proved that under
such conditions for $\alpha$ the map \eqref{eq:43} is an isomorphism between
$H'^{2}_{-1/2}$ and $L'^{2}_{-5/2}$.

We have satisfied all the hypothesis of the Implicit Function
Theorem. Hence, there exists a neighborhood $W=(-\lambda_0,\lambda_0)$ of the
origin in $\R$ such that the conclusion of theorem \ref{teo} holds.
\end{proof}

\textbf{Remarks:} We have imposed the perturbation functions $\omega$ and
$q$ to have compact support. This can be relaxed by requiring appropriate fall off
conditions at the axis and at the origin. 

The axially symmetric data considered here are not the most general one, since
we are assuming in the form of the metric \eqref{eq:105} that the axial Killing
vector is hypersurface orthogonal on the surface $S$ (but, of course, has a non
zero twist in the spacetime). This simplification allows one to use the
explicit expression's \eqref{eq:42b} for the second fundamental form. We expect
that this result can be generalized without this assumption. However, it is
important to emphasize that given a data as the one constructed in this
theorem, the time evolution described in section \ref{secpert}, under the
condition for lapse and shift \eqref{eq:5}--\eqref{eq:10}, will develope
initial data with the same asymptotic geometry for which the Killing vector is
not surface orthogonal. And hence we get from our family also non-trivial
initial data for which the Killing vector is not hypersurface orthogonal.

\section{Final Comments}
\label{sec:final-comments}

We have prove the existence of a initial data family close to extreme Kerr
black hole initial data. This family represent the natural initial data to
study the evolution near an extreme black hole in axial symmetry, in the spirit
of \cite{Dain:2008xr} \cite{Dain:2009wi}.

There exists also relevant open problems that can be address at the level of
the initial data. As we have seen in section \ref{secpert}, the extreme Kerr
black data lies outside the black hole region and hence they contain no trapped
surfaces. Does the family $(S,h_{ij}(\lambda), K_{ij}(\lambda))$ contains
trapped surfaces for $\lambda>0$? If these data have no trapped surfaces, then
there is a chance that they also lies outside the black hole region. This can,
of course, only been answered after the whole evolution has been analyzed. On
the other hand, if there are trapped surfaces, then the data necessarily
penetrate the black hole.  The formation of trapped surfaces for arbitrary
small $\lambda>0$ will indicate that extreme Kerr data is a very special
element in the family $(S,h_{ij}(\lambda), K_{ij}(\lambda))$.  In that case,
these kind of data could be very useful in the study of geometric inequalities
which relates angular momentum and area of trapped surfaces (see section 8 in
the review article \cite{Mars:2009cj}).

\section*{Acknowledgments}
S. D. thanks Piotr Chru\'sciel and Raffe Mazzeo for useful discussions. The
discussions with P. Chru\'sciel took place at the Institut Mittag-Leffler,
during the program ``Geometry, Analysis, and General Relativity'', 2008
fall. The discussions with R. Mazzeo took place at the Mathematisches
Forschungsinstitut Oberwolfach during the workshop ``Mathematical Aspects of
General Relativity'', October 11th -- October 17th, 2009.  S. D. thanks the
organizers of these events for the invitation and the hospitality and support
of the Institut Mittag-Leffler and Mathematisches Forschungsinstitut
Oberwolfach.

The authors want to thank Robert Beig for useful discussions that took place at
FaMAF during his visit in 2009.

S. D. is supported by CONICET (Argentina). M. E. G. C. is supported by a
fellowship of CONICET (Argentina). This work was supported in part by grant PIP
6354/05 of CONICET (Argentina), grant 05/B415 Secyt-UNC (Argentina) and the
Partner Group grant of the Max Planck Institute for Gravitational Physics,
Albert-Einstein-Institute (Germany).
\appendix

\section{Weighted Sobolev spaces}
\label{sec:weight-sobol-spac} 
The Bartnik's weighted Sobolev spaces $W'^{k,p}_\delta$  \cite{Bartnik86} are
appropriate for studying geometries with one cylindrical and one asymptotically
flat end. These functional spaces have weights both at infinity and at the
origin. 

The weighted Lebesgue spaces $L'^p_\delta$ are defined as the completion of
$\cs{\Rn\setminus\{0\}}$ functions under the norms
\begin{equation}
\label{eq:35}
\Vert f\Vert'_{p,\delta}=\left(\int_{\Rt\setminus\{0\}} |f|^pr^{-\delta p-n}dx\right)^{1/2}.
\end{equation}
The weighted Sobolev spaces  $W'^{k,p}_\delta$ are defined in the usual way
\begin{equation}
\label{eq:38}
\Vert f\Vert'_{k,p,\delta} =\sum_{0}^{m}\Vert D^jf\Vert_{p,\delta-j}.
\end{equation}
In this article we only use the cases $n=3$ and $p=2$, we have denoted these
spaces by $H'^k_\delta=W'^{k,2}_\delta$ and the norms by $ \Vert
f\Vert_{L'^2_\delta}= \Vert f\Vert'_{2,\delta}$ and $ \Vert
f\Vert_{H'^k_\delta}= \Vert f\Vert'_{k,2,\delta}$. 

The next lemma plays a crucial role in the proof of theorem \ref{teo}. 

\begin{lemma}\label{decaimiento}
Assume $u\in W'^ {k,p}_{\delta}$ with  $n-kp<0$, then we have the following estimate
 \begin{equation}
   \label{eq:41}
   r^{-\delta}|u|\leq C \left\lVert u  \right\rVert'_{k,p,\delta}. 
 \end{equation}
Moreover, we have
 \begin{equation}
\label{eq:44}
 \lim_{r\to 0}r^{-\delta}  |u|= \lim_{r\to \infty}r^{-\delta}  |u|= 0. 
\end{equation}
\end{lemma}
We will use this lemma only for the particular case $p=2$, $n=3$, $k=2$ and
$\delta=-1/2$, we state however the proof for the general case since it can
have other applications.
\begin{proof}
  This proof is adapted from \cite{Bartnik86}, Theorem 1.2, where the statement
  is proved for weighted spaces at infinity  (namely, $W ^{k,p}_\delta$
  spaces in the notation of \cite{Bartnik86}). 

Let $B_R$ be the ball of radius $R$ centered at the origin, and
let $A_R$ be  the annulus  $A_R=B_{2R}\setminus B_R$. We define the rescaled function
\begin{equation}
\label{eq:49}
u_R(x):=u(Rx). 
\end{equation}
Then, the fundamental scaling property of the spaces $W'^{k,p}_\delta$
(cf. equation after equation (1.3) in \cite{Bartnik86})  is given by 
\begin{equation}
\label{eq:50}
\Vert u_R\Vert_{k,p,\delta;A_1}=R^\delta\Vert u\Vert_{k,p,\delta;A_R},
\end{equation}
where we have used the same notation  as in  \cite{Bartnik86} for norms over
subsets of $\Rn$. 

We have
\begin{align}
  \label{eq:48}
  \sup_{A_R}r^{-\delta}|u| & =\sup_{A_1}R^{-\delta}r^{-\delta}|u_R|,\\
&\leq C R^{-\delta}\Vert r^{-\delta}u_R\Vert_{k,p;A_1}, \label{eq:48b}\\
&\leq CR^{-\delta} \Vert u_R\Vert'_{k,p,\delta;A_1},\label{eq:48c} \\
&= C \Vert u\Vert'_{k,p,\delta;A_R}.\label{eq:48d}
\end{align}
The line \eqref{eq:48} is a trivial change of coordinates. For the inequality
\eqref{eq:48b} we have used the standard Sobolev estimate on the bounded domain
$A_1$, which is valid for $n-kp<0$. We have denoted the standard Sobolev norm
on a domain $\Omega$ by $\Vert \cdot\Vert_{k,p;\Omega}$.  It is important to
note that the constant $C$ does not depend on $R$, since the domain $A_1$ does
not either. The inequality in \eqref{eq:48c} is trivial because on the domain
$A_1$ the two norms (standard and weighted) are equivalent. Finally, in
\eqref{eq:48d} we applied the scaling property \eqref{eq:50}.

Consider the set of annulus $A_{2^j}$ and define $u_j=u|_{A_{2^j}}$. It is
clear that 
\begin{equation}
  \label{eq:51}
  u=\sum_{j=-\infty}^{\infty}u_j.
\end{equation}
Then, we use the estimate \eqref{eq:48} on $A_{2^j}$ and sum over all $j$
\begin{align}
  \label{eq:52}
(\sup r^{-\delta}|u|)^p  \leq  \sum_{j=-\infty}^{\infty}
(\sup r^{-\delta}|u_j|)^p & \leq
 C\sum_{j=-\infty}^{\infty} \Vert u_j\Vert'^p_{k,p,\delta},\\
& = C  \Vert u \Vert'^p_{k,p,\delta}.
\end{align}
which proves \eqref{eq:41}. 

To prove \eqref{eq:44} we observe that the sum $  \sum_{j=-\infty}^{\infty}
(\sup r^{-\delta}|u_j|)^p $ is an infinite sum of positive real numbers which
is bounded, hence in the limit we must have
\begin{equation}
  \label{eq:53}
 \lim_{j\to \pm\infty} (\sup r^{-\delta}|u_j|)=0,
\end{equation}
which is equivalent to \eqref{eq:44}. 
\end{proof}

\section{Properties of extreme Kerr initial data}
\label{sec:appendix}
The spacetime metric for extreme Kerr black hole  in
Boyer-Lindquist coordinates $(t,\tilde{r},\theta,\phi)$, is given by
\begin{equation}
\label{eq:69}
  g=-\frac{\Delta\sin^2\theta}{\eta}dt^2+\eta(d\phi-\Omega
  dt)^2+\frac{\Sigma}{\Delta}d\tilde r^2+\Sigma d\theta^2 
\end{equation}
where $\eta$ is the square norm of the axial Killing vector 
\begin{equation}
\label{eq:59}
\eta^ \mu=\left(\frac{\partial}{\partial\phi}\right)^ \mu, \quad
\eta =g_{\nu\mu}\eta^\mu\eta^\mu,
\end{equation}
given by
\begin{equation}
  \label{eq:59b}
 \eta =\frac{(\tilde
  r^2+a^2)^2-a^2\Delta\sin^2\theta}{\Sigma}\sin^2\theta. 
\end{equation}
The functions $\Delta$ and $\Sigma$ are given by
\begin{equation}
 \Delta=(\tilde{r}-m)^2,\quad \Sigma=\tilde{r}^2+a^2\cos^2\theta,
\end{equation}
and $\Omega$ is the angular velocity 
\begin{equation}
 \Omega=\frac{2a^2\tilde r\sin^2\theta}{\eta\Sigma}.
\end{equation}
Here $a=J/m$ is the angular momentum per unit mass, and we consider the extreme
case $\sqrt|J|=m$. Note that for extreme Kerr we have two possible values for the
angular momentum $J=\pm m^2$ (and hence $a=\pm m$). 

Take a surface $t=constant$, define the  radius $r$ as  $r=\tilde{r}-m$. From 
\eqref{eq:69} we deduce that the  intrinsic metric on this surface has the form
\eqref{thcero} with 
\begin{equation}\label{ficero}
 e^{2q_0}=\frac{\Sigma\sin^2\theta}{\eta}, \quad
 \Phi_0^4=\frac{\eta}{\rho^2}. 
\end{equation}
The twist potential of the Killing vector $\eta^\mu$ is given by 
\begin{equation}\label{wcero}
\omega_0 =2J(\cos^3\theta-3\cos\theta)-\frac{2Jm^2\cos\theta\sin^4\theta}{\Sigma}.
\end{equation}
 The lapse function and shift vector for this foliation are given by
\begin{align}
\label{eq:65}
  \alpha &=\frac{r}{\sqrt{\Sigma+a^2(1+2a(r+a)/\Sigma)\sin^2\theta}},\\
\label{eq:66}
 \beta^\phi &=-\frac{2a^2\sin^2\theta(r+a)}{\Sigma^3}r^2.
\end{align}
The following asymptotic limit are interesting
\begin{align}
  \label{eq:56}
  \lim_{r\to0} \sqrt{r}\Phi_0 &= \left(\frac{4m^2}{1+\cos^2\theta}\right)^{1/4},\\
  \label{eq:57}
   \lim_{r\to0}  e^{2q_0} &= \left(\frac{1+\cos^2\theta}{2}\right)^{2},\\
  \label{eq:58}
   \lim_{r\to0} \omega_0 &=-\frac{8J\cos\theta}{1+\cos^2\theta},\\
 \lim_{r\to0} R  & = \frac{2\sin^2\theta}{m^2(1+\cos^2\theta)^3}. \label{eq:58b}
\end{align}
We take the opportunity to correct a misprint in equation A.15 of
\cite{Avila:2008te}. There is a missing exponent $3$ in the denominator of this
formula, it should be the same as equation \eqref{eq:58b}.
 
In the following, we 
use $\Delta$ to denote the flat Laplace operator in three dimensions, the two
dimensional Laplacian $\Delta_2$ is given by
\begin{equation}
  \label{eq:67}
 \Delta_2=  \frac{1}{r}\partial_r(r\partial_r )+ \frac{1}{r^2}\partial^2_\theta . 
\end{equation}

The next lemma plays a crutial role in the proof of theorem \ref{teo}. 
\begin{lemma}
\label{lemaalpha}
  Let $q_0$ and $\Phi_0$ be given by \eqref{ficero} and $\omega_0$ by
  \eqref{wcero}. Then the function $\alpha$ defined in \eqref{alpha},
has the form $\alpha=hr^{-2}$ where $h\geq 0$ and $h$ is bounded in $\Rt$.
\end{lemma}
\begin{proof}
From the Hamiltonian constraint
\begin{equation}
-\frac{\Delta_2q_0}{4}=\frac{\Delta\Phi_0}{\Phi_0}+\frac{(\partial w_0)^2}{16\eta^2}.
\end{equation}
and the stationary equation satisfied by extreme Kerr's initial data (see \cite{Avila:2008te})
\begin{equation}
\frac{\Delta\Phi_0}{\Phi_0}=-\frac{(\partial
  \omega_0)^2}{4\eta^2}+\frac{(\partial\Phi_0)^2}{\Phi_0^2} 
\end{equation}
we obtain
\begin{equation}
-\frac{\Delta_2q_0}{4}=-\frac{3}{16}\frac{(\partial
  w_0)^2}{\eta^2}+\frac{(\partial \Phi_0)^2}{\Phi_0^2}. 
\end{equation}
Therefore
\begin{equation}\label{alpha1}
\alpha=\frac{(\partial \omega_0)^2}{4\eta^2}+\frac{(\partial \Phi_0)^2}{\Phi_0^2},
\end{equation}
which is clearly a non negative quantity. By an explicit calculation it can be
seen that $\alpha$ is in fact a strictly positive function. Since we do not need
this property for our purposes, we omit the details.
 Also by explicit means, we note that $\alpha$ is $O(r^{-2})$
at the origin, and $O(r^{-4})$ at infinity, being otherwise bounded. Thereby,
there must exist a positive function $h$ such that $\alpha=hr^{-2}$.
\end{proof}
It is important to note that in the proof of lemma \ref{lemaalpha} we have used
the fact that extreme Kerr satisfies the stationary Einstein equations and also
that the topology of extreme Kerr allows us to choose these coordinates. In
particular, the proof fails for non-extreme Kerr. See a  similar discussion
in  \cite{Dain05c} at the end of page 6868.

\begin{lemma}\label{desi}
  Let $\Phi_0$, $q_0$ and $\omega_0$ be defined by \eqref{ficero} and \eqref{wcero}, and
assume that $m>0$. Then we have the following bounds:

\begin{align}\label{cotasficero}
 \sqrt{m}& \leq \sqrt{r+m}  \leq \sqrt{r} \Phi_0\leq \sqrt{2} \sqrt{r+m} ,\\
\label{eq:cotaw0}
\frac{(\partial \omega_0)^2}{\rho^4} &\leq116\frac{m^4}{r^6},\\
\label{cotaqcero}
|\Delta_2 q_0|& \leq\frac{90}{r^2}.
\end{align}
\end{lemma}

\begin{proof}

 Inequality  \eqref{cotasficero}  has been proved in \cite{Avila:2008te} (see equations
   (10) and (12) in this reference). 

 We have
\begin{equation}\label{partialwcero}
(\partial \omega_0)^2=\frac{4m^4\rho^6F}{r^8\Sigma^4}
\end{equation}
where
\begin{equation}
F=4r^2a^4\tilde r^2\sin^2(2\theta)+\left(3\tilde r^4+a^2\tilde r^2+a^2(\tilde
  r^2-a^2)\cos^2\theta\right)^2 
\end{equation}
and $\tilde r=r+m$. Then
\begin{eqnarray}
F&\leq&4r^2a^4\tilde r^2+\left(3\tilde r^4+a^2\tilde r^2+a^2\tilde
  r^2\right)^2\leq29(r+a)^8. 
\end{eqnarray}
We also find, bounding $\Sigma\geq(r+a)^2$ and $\rho\leq r$, that
\begin{equation}
(\partial \omega_0)^2\leq\frac{4a^4\rho^4 29(r+a)^8}{r^6(r+a)^8}=116\frac{m^4}{r^2}.
\end{equation}

Finally, using the explicit expressions for $\Phi_0$ and $\omega_0$ one can
check, after a laborious but straightforward calculation, the bound on
$|\Delta_2q_0|$.

\end{proof}

\section{The implicit function theorem}
\label{sec:impl-funct-theor}
To facilitate the readability of the article and also to fix the notation, we
reproduce in this appendix well known results on differential calculus in
Banach spaces (see, for example \cite{abraham88}, \cite{Choquet77}, and also
the more introductory text books \cite{McOwen96}, \cite{Renardy04}).

Let $X$ and $Z$ be Banach spaces. Let $A:X\to Z$ be a linear bounded operator.
We denote by $\mathcal{L}(X,Z)$ the set of all linear and bounded operators
from $X$ to $Z$. The set $\mathcal{L}(X,Z)$ is itself a Banach space with the
operator norm defined by
\begin{equation}
  \label{eq:46}
\lVert A \rVert_{\mathcal{L}(X,Z)}=\sup_{\lVert x \rVert \neq 0}
\frac{\lVert A(x) \rVert_Z}{\lVert x \rVert_X}. 
\end{equation}
Let $x$ be a point in $X$ and let $G$ be a mapping from a neighborhood of $x$ into
$Z$. Then $G$ is called Fr\'echet differentiable at the point $x$ if there exists
a linear operator $DG(x)\in \mathcal{L}(X,Z)$ such that
\begin{equation}
  \label{eq:45}
\lim_{v\to0}\frac{\lVert G(x+v)-G(x)-DG(x)[v]\rVert}{\lVert x \rVert_X  }=0. 
\end{equation}
The map $G$ is called continuously differentiable (i.e. $C^1$) if the derivative
$DG(x)$ as an element of $\mathcal{L}(X,Z)$  depends continuously on
$x$. Namely,  for every $\delta>0$ there
  exists $\epsilon>0$ such that
\begin{equation}
  \label{eq:47}
\lVert x_1-x_2 \rVert_X <\epsilon \Rightarrow \lVert D G(x_1)- DG(x_2)
\rVert_{\mathcal{L}(X,Z)}   <\delta.
  \end{equation}
  Let $X$, $Y$ and $Z$ be Banach spaces and let $G$ be a map $G:X\times Y \to
  Z$, in a similar way we define the partial derivatives with respect to the
  first argument by $D_1G(x,y)$ and with respect to the second argument by
  $D_2G(x,y)$.

\begin{theorem}[Implicit Function Theorem]
\label{t:ift}
Suppose $U$ is a neighborhood of $0$ in $X$, $V$ is a neighborhood of $0$ in
$Y$, and $G:X\times Y \to Z$ is $C^1$. Suppose $G(0,0)=0$ and $D_2G(0,0) :Y\to
Z$ defines a bounded operator and it is an isomorphism.  Then, there exists a
neighborhood $W$ of the origin in $X$ and a continuously differentiable mapping
$f:W\to Y$ such that $G(x,f(x))=0$. Moreover, for small $x$ and $y$, $f(x)$ is
the only solution $y$ of the equation $G(x,y)=0$.
\end{theorem}


\begin{thebibliography}{10}

\bibitem{abraham88}
R.~Abraham, J.~E. Marsden, and T.~Ratiu.
\newblock {\em Manifolds, tensor analysis, and applications}, volume~75 of {\em
  Applied Mathematical Sciences}.
\newblock Springer-Verlag, New York, second edition, 1988.

\bibitem{Avila:2008te}
G.~A. Avila and S.~Dain.
\newblock {The Yamabe invariant for axially symmetric two Kerr black holes
  initial data}.
\newblock {\em Class. Quantum. Grav.}, 25:225002, 2008, gr-qc/00805.2754.

\bibitem{Physics.2.102}
V.~Balasubramanian.
\newblock Are black holes really two dimensional?
\newblock {\em Physics}, 2:102, Dec 2009.

\bibitem{Bardeen:1999px}
J.~M. Bardeen and G.~T. Horowitz.
\newblock {The extreme Kerr throat geometry: A vacuum analog of AdS(2) x S(2)}.
\newblock {\em Phys. Rev.}, D60:104030, 1999, hep-th/9905099.

\bibitem{Bartnik86}
R.~Bartnik.
\newblock The mass of an asymptotically flat manifold.
\newblock {\em Comm. Pure App. Math.}, 39(5):661--693, 1986.

\bibitem{Boyer67}
R.~H. Boyer and R.~W. Lindquist.
\newblock Maximal analytic extension of the {K}err metric.
\newblock {\em J. Math. Phys.}, 8(2):265--281, 1967.

\bibitem{Carter:1968rr}
B.~Carter.
\newblock {Global structure of the Kerr family of gravitational fields}.
\newblock {\em Phys. Rev.}, 174:1559--1571, 1968.

\bibitem{Carter73}
B.~Carter.
\newblock Black hole equilibrium states.
\newblock In {\em Black holes/Les astres occlus (\'Ecole d'\'Et\'e Phys.
  Th\'eor., Les Houches, 1972)}, pages 57--214. Gordon and Breach, New York,
  1973.

\bibitem{Choquet77}
Y.~Choquet-Bruhat, C.~de~Witt-Mortte, and M.~Dillard-Bleick.
\newblock {\em Analysis, Manifolds and Physics}.
\newblock North-Holland, Amsterdam, 1977.

\bibitem{Chrusciel:2007ak}
P.~T. Chru{\'s}ciel, Y.~Li, and G.~Weinstein.
\newblock Mass and angular-momentum inequalities for axi-symmetric initial data
  sets. {II}. {A}ngular-momentum.
\newblock {\em Ann. Phys.}, 323(10):2591--2613, 2008, arXiv:0712.4064.

\bibitem{Chrusciel:2008js}
P.~T. Chrusciel and J.~Lopes~Costa.
\newblock {On uniqueness of stationary vacuum black holes}, 2008, 0806.0016.

\bibitem{Chrusciel:2009ki}
P.~T. Chrusciel and J.~Lopes~Costa.
\newblock {Mass, angular-momentum, and charge inequalities for axisymmetric
  initial data}.
\newblock {\em Class. Quant. Grav.}, 26:235013, 2009, 0909.5625.

\bibitem{Costa:2009hn}
J.~L. Costa.
\newblock {A Dain Inequality with charge}, 2009, 0912.0838.

\bibitem{Dain99b}
S.~Dain.
\newblock Initial data for a head on collision of two {K}err-like black holes
  with close limit.
\newblock {\em Phys. Rev. D}, 64(15):124002, 2001, gr-qc/0103030.

\bibitem{Dain05d}
S.~Dain.
\newblock Proof of the (local) angular momemtum-mass inequality for
  axisymmetric black holes.
\newblock {\em Class. Quantum. Grav.}, 23:6845--6855, 2006, gr-qc/0511087.

\bibitem{Dain05c}
S.~Dain.
\newblock A variational principle for stationary, axisymmetric solutions of
  einstein's equations.
\newblock {\em Class. Quantum. Grav.}, 23:6857--6871, 2006, gr-qc/0508061.

\bibitem{Dain:2008xr}
S.~Dain.
\newblock {Axisymmetric evolution of Einstein equations and mass conservation}.
\newblock {\em Class. Quantum. Grav.}, 25:145021, 2008, 0804.2679.

\bibitem{Dain:2007pk}
S.~Dain.
\newblock The inequality between mass and angular momentum for axially
  symmetric black holes.
\newblock {\em International Journal of Modern Physics D}, 17(3-4):519--523,
  2008, arXiv:0707.3118 [gr-qc].

\bibitem{Dain06c}
S.~Dain.
\newblock Proof of the angular momentum-mass inequality for axisymmetric black
  holes.
\newblock {\em J. Differential Geometry}, 79(1):33--67, 2008, gr-qc/0606105.

\bibitem{Dain:2008yu}
S.~Dain and M.~E. Gabach~Cl\'ement.
\newblock {Extreme Bowen-York initial data}.
\newblock {\em Class. Quantum. Grav.}, 26:035020, 2009, 0806.2180.

\bibitem{Dain:2008ck}
S.~Dain, C.~O. Lousto, and Y.~Zlochower.
\newblock {Extra-Large Remnant Recoil Velocities and Spins from Near-
  Extremal-Bowen-York-Spin Black-Hole Binaries}.
\newblock {\em Phys. Rev. D}, 78:024039, 2008, 0803.0351.

\bibitem{Dain:2009wi}
S.~Dain and O.~E. Ortiz.
\newblock {On well-posedness, linear perturbations and mass conservation for
  axisymmetric Einstein equation}, 2009, 0912.2426.

\bibitem{gabach09}
M.~E. Gabach~Cl\'ement.
\newblock {Conformally flat black hole initial data, with one cylindrical end},
  2009, 0911.0258.

\bibitem{Guica:2008mu}
M.~Guica, T.~Hartman, W.~Song, and A.~Strominger.
\newblock {The Kerr/CFT Correspondence}.
\newblock {\em Phys. Rev. D}, 80:124008, 2009, 0809.4266.

\bibitem{Hannam:2009ib}
M.~Hannam, S.~Husa, and N.~O. Murchadha.
\newblock {Bowen-York trumpet data and black-hole simulations}.
\newblock {\em Phys. Rev.}, D80:124007, 2009, 0908.1063.

\bibitem{Hannam:2008sg}
M.~Hannam, S.~Husa, F.~Ohme, B.~Bruegmann, and N.~O'Murchadha.
\newblock {Wormholes and trumpets: the Schwarzschild spacetime for the
  moving-puncture generation}.
\newblock {\em Phys. Rev.}, D78:064020, 2008, 0804.0628.

\bibitem{Immerman:2009ns}
J.~D. Immerman and T.~W. Baumgarte.
\newblock {Trumpet-puncture initial data for black holes}.
\newblock {\em Phys. Rev.}, D80:061501, 2009, 0908.0337.

\bibitem{Mars:2009cj}
M.~Mars.
\newblock {Present status of the Penrose inequality}.
\newblock {\em Class. Quant. Grav.}, 26:193001, 2009, 0906.5566.

\bibitem{mcclintock06}
J.~E. McClintock, R.~Shafee, R.~Narayan, R.~A. Remillard, S.~W. Davis, , and
  L.-X. Li.
\newblock The spin of the near-extreme kerr black hole grs 1915+105.
\newblock {\em The Astrophysical Journal}, 652(1):518--539, 2006.

\bibitem{McOwen96}
R.~C. McOwen.
\newblock {\em Partial Differential Equation}.
\newblock Prentice Hall, New Jersey, 1996.

\bibitem{Reid:2008fp}
M.~J. Reid.
\newblock {Is there a Supermassive Black Hole at the Center of the Milky Way?}
\newblock {\em Int. J. Mod. Phys.}, D18:889--910, 2009, 0808.2624.

\bibitem{Renardy04}
M.~Renardy and R.~C. Rogers.
\newblock {\em An introduction to partial differential equations}, volume~13 of
  {\em Texts in Applied Mathematics}.
\newblock Springer-Verlag, New York, second edition, 2004.

\end{thebibliography}

\end{document}